\documentclass[prl,twocolumn,aps,superscriptaddress]{revtex4}
\usepackage[a4paper, total={6.5in, 10in}]{geometry}
\usepackage[utf8]{inputenc}
\usepackage{mathtools}
\usepackage{xcolor}
\usepackage{enumitem}
\usepackage{amsthm}
\usepackage{wrapfig}
\usepackage{graphicx}
\usepackage{enumitem}
\usepackage{float}
\usepackage[french, english]{babel}
\setlist{nosep} % or \setlist{noitemsep} to leave space around whole list
\usepackage{physics}
\usepackage{hyperref} %This package must always be the last

\usepackage{algorithm}
\floatname{algorithm}{Protocol}
\usepackage{algorithmic}
\newtheorem{theorem}{Theorem}

\newtheorem{lemma}{Lemma}

\newtheorem{manualtheoreminner}{Theorem}
\newenvironment{manualtheorem}[1]{%
  \renewcommand\themanualtheoreminner{#1}%
  \manualtheoreminner
}{\endmanualtheoreminner}

\newcounter{protocol}

\newcommand{\fede}[1]{\textcolor{red}{#1}}

\begin{document}
\title{Quantum protocol for electronic voting without election authorities}%e-voting without election authorities}
%\date{\vspace{-5ex}}
\author{Federico Centrone}
\email{fcentrone@icfo.net}
\affiliation{Sorbonne Université, CNRS, LIP6, 4 place Jussieu, F-75005 Paris, France}\affiliation{Université de Paris, CNRS, IRIF, 8 Place Aurélie Nemours, 75013 Paris, France}
\author{Eleni Diamanti}
\email{eleni.diamanti@lip6.fr}
\affiliation{Sorbonne Université, CNRS, LIP6, 4 place Jussieu, F-75005 Paris, France}
\author{Iordanis Kerenidis}
\email{jkeren@irif.fr}
\affiliation{Université de Paris, CNRS, IRIF, 8 Place Aurélie Nemours, 75013 Paris, France}
\date{\today}

\begin{abstract}
    Electronic voting is a very useful but challenging internet-based protocol that despite many theoretical approaches and various implementations with different degrees of success, remains a contentious topic due to issues in reliability and security.
    %In fact, the very definition of security in electronic voting is ambiguous, which explains why almost all the proposed schemes were subsequently declared insecure.
    Here we present a quantum protocol that exploits an untrusted source of multipartite entanglement to carry out an election without relying on election authorities, simultaneous broadcasting or computational assumptions, and whose result is publicly verifiable. The level of security depends directly on the fidelity of the shared multipartite entangled quantum state, and the protocol can be readily implemented for a few voters with state-of-the-art photonic technology.
\end{abstract}
\maketitle

\section{Introduction}

Electronic voting, or e-voting, is a functionality built on top of the Internet or any distributed network that allows performing large-scale elections in a secure and verified way, even in the presence of distrusted authorities or dishonest agents. The benefits of such a functionality include a faster and simpler way to carry out elections resulting in higher public participation (\emph{i.e.}, a higher number of voters), reduction of election costs, and accessibility for people with disabilities. Furthermore, e-voting offering information-theoretic security guarantees in principle the security and honesty of the elections even in the case of corrupted officials or a coalition of dishonest agents. However, the adoption of a protocol that uses a public network to accomplish elections also increases the possibilities for fraud by manipulating the results or violating privacy~\cite{hagai2015}. Moreover, even though it may not be possible for such a protocol to be infringed, the agents would need to trust devices and programs they did not author and, most likely, not even understand~\cite{thompson1984}.  Finally, it is also necessary to take into account the cost of implementing the elections with advanced technology.

Classical e-voting systems are based on computational assumptions and might not be secure against quantum or other adversaries. Moreover, there have been serious criticisms against commercial e-voting systems due to insecurities~\cite{abba2017}. In recent years, several quantum e-voting protocols have been proposed, announcing perfect security also in dishonest scenarios. However, none of these was able to provide a rigorous mathematical definition of the properties required, such as privacy, verifiability, and correctness, as well as to identify proper corruption models suitable for this scenario. As a matter of fact, in~\cite{arapinis2021}, the authors discovered vulnerabilities in all previously proposed quantum e-voting schemes. Let us also mention the work in~\cite{chilotti2016}, where a lattice based post-quantum cryptographic protocol achieving computational security was suggested. This may however be undesirable for e-voting because privacy cannot be guaranteed in the long term. For these reasons it is paramount to find schemes based on information-theoretic security, rather than computational assumptions, in order to ensure honest elections also in the presence of dishonest authorities with unbounded (or much bigger than publicly known) computational power. %In other words, we aim for e-voting schemes with \textit{information-theoretic security}.
This level of security for an e-voting scheme was announced in~\cite{broadbent2007}, which proposed a protocol exploiting only classical resources. However, the requirement of a simultaneous broadcasting channel makes it impractical even for a small number of voters, or turns the security back to computational if the simultaneous broadcasting channel is simulated via usual channels.

Here we describe and formalize several properties required by an electronic voting system to be secure and propose a quantum protocol that satisfies these properties even in the presence of computationally unbounded adversaries and without necessarily trusting the devices that execute the elections. Another benchmark is that of practicality, in the sense that we want the protocol to be implementable with technology that is already or soon-to-be available and hence that it is possible to carry out a demonstration at least for a few voters. Our protocol fulfils the above requirements, at the expense on relying on the generation and manipulation of a Greenberger-Horne-Zeilinger (GHZ) state with as many particles as voters, which is the major limitation to its scalability. We note, however, that here the number of voters can refer not to the total number of voters in the election, but the number of voters within each polling station, since, as in the classical case, we can aspire to provide privacy of each vote within each such polling station.

Our protocol utilises a multipartite entanglement verification scheme~\cite{pappa2012, mccutcheon2016, yehia2021} as a subroutine as well as classical subroutines useful for anonymous transmission in communication networks. It is inspired by the self tallying quantum anonymous voting protocol proposed in~\cite{wang2016}, the particularity of which resides in the absence of a tallier and any election authority. Although this protocol was proven insecure in~\cite{arapinis2021}, by employing the multipartite entanglement verification scheme of~\cite{pappa2012} and simplifying the quantum resource requirements using ideas of~\cite{broadbent2007}, we devise an efficient quantum protocol and rigorously prove its security. Furthermore, even though sharing an $N$-party GHZ state for big $N$ needed for large-scale elections is today still technologically out of reach, we note that applications for small number of voters are already feasible and important, and so are applications where one can use multiple small-scale GHZ states to mimic an election with a large number of polling stations. Similarly to the proposal of~\cite{wang2016}, such a protocol can also be used as an anonymous chat board, where one party can write a message visible to anyone but no one can deduce who sent it (similar to one party being able to vote without anyone being able to deduce whose vote that is), or even as a form of anonymous distributed computation. 

%%%%%%%%%%%%%%%%%%%%%%%%%%%%%%%%%

\section{Quantum e-voting protocol}

%We start with a high level description of our quantum e-voting protocol and its components, before analyzing it in detail.
In the general setting of our protocol, a source of $N$-qubit GHZ states, $\ket{GHZ}=\tfrac{1}{\sqrt{2}}(\ket{0}^{\otimes N}+\ket{1}^{\otimes N})$, is situated at the central node of a star-graph quantum network, whose edges are the communication links needed for the distribution of the entangled qubits to $N$ agents. Even though voters do not have to trust the multipartite entangled photon source, it should be capable of producing high fidelity quantum states to pass the verification test at the heart of the protocol in the honest case, and with a high enough rate to ensure the elections can be performed in an efficient way.
%The central node that connects to all agents is called a \textit{Qonnector} and is equipped with a multipartite quantum source.
Each agent only needs to be able to receive, perform unitary operations, store for a short time, and measure single photons.
%, called a \textit{Qlient}, is situated in one of the peripheral nodes and can receive, store for a short time, perform a unitary operation, and measure a single qubit/photon. \iordanis{add a reference to the QCity manuscript}
A particular feature of our protocol is that it does not require talliers or other election authorities, as all the votes are announced publicly and anonymously and so each voter can verify that the tally is correct. However, as we will see later, if a voter detects malicious behaviour, they can abort the protocol using the appropriate subroutine at the end of the election.

Let us now describe the protocol, referring to a number of classical and quantum subroutines when it is necessary. The pseudo-code of each of the subroutines is provided in Appendix A.
We will assume in the following that the election admits only two possible candidates, `0' and  `1', while the generalization to additional candidates is described later.

In the first phase of the protocol, each agent $k \in [N]$ needs to obtain a secret, unique index $\omega_k \in [N]$ that indicates the round the agent becomes the voting agent. To do this, the agents perform the \textsf{UniqueIndex} subroutine.

Subsequently, the second phase consists of as many rounds as the voters and at each round one agent votes according to the order based on the secret indices shared in the first phase of the protocol.

Each voting round $\ell \in [N]$ starts with 
%starts with the untrusted (or just faulty) source distributing the multipartite quantum state, and then 
the voting agent (namely the agent $k$ who has received the unique index $\omega_k=\ell$) deciding repeatedly to perform one of two actions according to some random coins they flip locally: \textsf{Verification} of the source or \textsf{Voting}. The probability of this decision is guided by a %security 
parameter $M$, which equals the number of coins, so that the probability the coins return `all heads' (which corresponds to \textsf{Voting}) is $2^{-M}$. In order to notify everyone anonymously of the outcome of the coin flip, all agents then perform a \textsf{LogicalOr} subroutine with input 0 except the voting agent whose input depends on the result of the coin flip: if the result was not `all heads' the agent inputs 1, announcing anonymously \textsf{Verification} to the other agents, otherwise the agent inputs 0 announcing \textsf{Voting}. For the \textsf{LogicalOr} protocol performed in this phase, we will assume for simplicity that if the voting agent inputs 1, then the probability that the outcome is 1 is equal to 1, which corresponds to the choice of a very small security parameter for this subroutine (see Lemma 2 in Appendix A).

%For the \textsf{LogicalOr} protocol performed in this phase, we can choose a very small security parameter $2^{-50}$, so that if the agent inputs 1, then the probability the outcome is 1 is at least $1-2^{-50}$, which for simplicity we will consider equal to 1 in the following.

When \textsf{Verification} is announced, following the corresponding protocol, the voting agent first performs the \textsf{RandomAgent} subroutine in order to choose a verifier anonymously; this is necessary because the verifier needs to communicate publicly with the other agents so if their identity is the same as the one of the voting agent, the voter's privacy would be violated. Then, they all proceed with the \textsf{Verification} test of the multipartite quantum state distributed by the untrusted (or just faulty) source.
In the ideal case, where the quantum state is created and distributed with no errors and all the operations are perfect, if the state does not pass the \textsf{Verification} test the protocol is aborted. In any realistic implementation, however, the protocol cannot abort as soon as there is any error. In practice, 
%if \textsf{Verification} fails a number of times, we can use this information to deduce the average trace distance of the state being used with respect to the ideal state and only if it is above a predetermined threshold we reject.  
at each voting round, during the verification tests before \textsf{Voting}, each honest
agent $j$ counts the number of trials and rejections when they are the verifier, computes the practical parameter $\delta_j=\frac{\text{rejections}_j}{\text{trials}_j}$, and if this is larger than a predetermined threshold $\delta$, the entire protocol is aborted.

When \textsf{Voting} is announced, the agents proceed with the corresponding subrtoutine, which returns an $N$-dimensional binary vector encoding the voting agent's preference. The underlying idea here is that if all qubits of the shared GHZ state are measured in the Hadamard basis, the sum of the outcomes $d_k$ modulo 2 of all agents is always zero. Then, at this round, all agents will just perform a Hadamard measurement on their qubit, while the voting agent $k$ will XOR the outcome of the Hadamard measurement with their vote intention $v_k$. This implies that when everyone follows the protocol, the parity of all announced outcomes in the round is equal to the vote intention $v_k$.

Then, a new round starts, the $(\ell+1)$-th round, where it is the turn of agent $k'$ with index $\omega_{k'}=\ell+1$ to be the voting agent.
After all voting rounds have completed and everyone has proceeded with \textsf{Voting}, all agents publicly broadcast their own result vectors and all together will form an $N\times N$ bulletin board $\mathbf{B}$. By computing the parity of each row (corresponding to each round) we get the vote vector $\mathbf{E}$ (since as we said the parity of each row is equal to the vote of the voting agent) from which the tally $\mathbf{T}$ can be calculated easily by everybody. Since the indices are unique and secret, each agent can verify that their vote is correct without revealing their choice. If an agent wants to abort the protocol because of suspected fraud (\emph{e.g.}, the tally does not agree with their vote intention) they can input their objection anonymously during the \textsf{LogicalOr} procedure that follows, where the security parameter defines how many agents on average should raise an objection before the election is actually aborted.

The pseudo code for the entire protocol is given below, while in Fig.~\ref{example} we provide a simple instance of the voting procedure.

\begin{figure*}
\[
d_1=\begin{pmatrix}
0 \\
1 \\
1 \\
0
\end{pmatrix},
d_2=\begin{pmatrix}
0 \\
0 \\
0 \\
1
\end{pmatrix},
d_3=\begin{pmatrix}
1 \\
1 \\
1 \\
1
\end{pmatrix},
d_4=\begin{pmatrix}
1 \\
0 \\
0 \\
0
\end{pmatrix} \longrightarrow
\mathbf{B}=\begin{pmatrix}
0 & 0 &  1 & \mathbf{1}\\
1 & \mathbf{1} & 1 & 0 \\
\mathbf{0} & 0 & 1 & 0 \\
0 & 1 & \mathbf{0} & 0
\end{pmatrix}\longrightarrow
\mathbf{E}=\begin{pmatrix}
0\\
1\\
1\\
1
\end{pmatrix}, \mathbf{T=\begin{pmatrix}
1\\
3\\
\end{pmatrix}}
\]
\caption{Example of the voting procedure according to our e-voting scheme with 4 voters who vote in the order $(4,2,1,3)$. At the end of all rounds each voter has a list of 4 Hadamard measurement outcomes $d_k$, and for each round the four outcomes sum to 0 modulo 2. The voters express their vote by adding their vote (0 or 1) to the row corresponding to their secret index (in bold), then broadcast the resulting vector and all together they form the bulletin board $\mathbf{B}$. Here the votes where $(0,1,1,1)$. Then they sum each row of $\mathbf{B}$ to compute the election vote set $\mathbf{E}$, from which is computed the tally $\mathbf{T}$. In this example candidate `1' won the election. } \label{example}
\end{figure*}

\begin{algorithm}[]
\caption{\textsf{Quantum e-voting}}
\begin{flushleft}
\textit{Input:}  
$N$ agent votes $\mathbf{V}=\{v_k\}_{k\in[N]}$, security parameter $S$ used in Phase 3, $\epsilon$: distance from the perfect GHZ state, $\delta$: threshold for verification, $\eta$: probability of failure of verification.\\
\textit{Output:} The candidate with majority votes or Abort.\\
\textit{Resources:} Classical communication, random numbers, $N$-qubit GHZ source, quantum channels.\\
\textit{Description: }
\end{flushleft}
\begin{algorithmic}[1]
\STATE Phase 1 [getting unique secret indices]:\\
    \begin{enumerate}
    \item Agents perform \textsf{UniqueIndex} until each one receives a secret unique random index $\omega_k$.
    \end{enumerate}
\STATE Phase 2 [casting votes]: \\
    For $\ell=1 \mbox{ to } N$ [voting round $\ell$]:
                \begin{enumerate}%[leftmargin=20pt]
                 \item [(a)] The voting agent is the agent $k$ with $\omega_k=\ell$.
                 \item [(b)] Repeat
                      \begin{enumerate}
                      \item [(i)] The source distributes to each of the $N$ agents one qubit of the GHZ state.
                      \item [(ii)] All agents $j \in [N]$ set $\text{rejections}_j=\text{trials}_j=0$;
                      \item [(iii)] The voting agent tosses $ \log_2\left[\frac{16N\epsilon^2}{(\epsilon^2-4\delta)^2}\ln{\left(\frac{1}{\eta}\right)}\right]$ coins;
                      \item [(iv)] The agents perform \textsf{LogicalOr}, where output 1 indicates \textsf{Verification} and output 0 indicates \textsf{Voting}, and where everyone except the voting agent inputs 0; 
                    if the coin toss is `all heads' the voting agent also inputs 0,  otherwise the voting agent  inputs 1;

                      \item [(v)] If  \textsf{Verification} is chosen, the agents perform \textsf{RandomAgent} and the voting agent picks anonymously an agent $j \in [N]$ to be the verifier. Agent $j$ increment $\text{trials}_j  $ by $1$ and if \textsf{Verification} outputs reject: agent $j$ increments $\text{rejections}_j$ by $1$.
                      \end{enumerate}
                 until \textsf{Voting} is announced.
                \item [(c)] If %for some $k$,
              for any $j \in [N]$ $\delta_j=\tfrac{\text{rejections}_j}{\text{trials}_j}>\delta$, the protocol Aborts. 
                \item [(d)]Perform \textsf{Voting}. The outcome is one row of the bulletin board $\mathbf{B}$.
                \end{enumerate}
\STATE Phase 3 [verification of results]: \\
    \begin{itemize}%[leftmargin=10pt]
    \item All agents perform \textsf{LogicalOr} with security parameter $S$, and with input 1 if their vote is not the same as the vote in the tally $\mathbf{T}$ for the round in which they were the voting agent, else with input 0.
    \item If \textsf{LogicalOr} outputs 1, Abort the protocol, else the candidate with the majority votes according to the tally wins the elections.
    \end{itemize}
\end{algorithmic}
\end{algorithm}

%\iordanis{What is $S$ in Phase 3?} \fede{$S$ should is bounded by $\sigma_H$ and $\sigma_D$ in the correctness property. In particular, if we specify the input parameters $\epsilon, N,  \sigma_H$ than $S$ is uniquely defined. }

%%%%%%%%%%%%%%%%%%%%%%%%%%%%%%%%%

\section{E-voting protocol Analysis}

We now analyze our quantum e-voting scheme and show that it possesses a number of desired properties even in the non-ideal case where the quantum source is imperfect or can be manipulated by colluding adversaries.

If the quantum states being used in the protocol are perfect GHZ states and the agents behave honestly, all the operations are anonymous and hence the e-voting scheme is perfectly correct and private. In any realistic scenario, however, the state used will have some imperfections, due to the source itself, the photon distribution, storage and measurement that may result in some errors in the tally, for example the sum of the outcomes of a round will not be 0 mod 2. We account for all the possible imperfections assuming that the fidelity between the state used in the protocol $\ket{\psi}$ and the perfect GHZ state $\ket{GHZ}$ is $F(\ket{\psi}, \ket{GHZ})= \sqrt{1-\epsilon^2}$ for some $\epsilon>0$. Note that the state produced by the source could be a mixed state, but as discussed in~\cite{unnikrishnan2019}, for the security it suffices to upper bound the cheating probability of any pure state, since this would also bound the cheating probability for any mixed state. For this reason we analyze below the case where the state produced by the source is a pure state.

Since the source or the state itself can further be intercepted and modified by an adversary in order to gain advantage over the privacy of the honest voters, we need to implement a mechanism that allows anyone to check the legitimacy of the state being used with high probability. An efficient multipartite entangled state verification protocol was devised in~\cite{pappa2012, mccutcheon2016} and applied to an anonymous transmission protocol in~\cite{unnikrishnan2019}. This is the \textsf{Verification} subroutine used in the protocol (see Appendix~A for details). While in the ideal case we would abort the protocol as soon as the test failed once, in a realistic implementation we need to keep track of the number of failures and at the end check if the failures are too many with respect to what was expected, %(due to the expected errors of the state produced by an honest but imperfect source),
which would imply that there was a malicious manipulation of the source. This is what is performed in Phase 2 of the protocol.

In~\cite{pappa2012}, the authors prove that the probability of a state $\ket{\psi}$, whose trace distance with the GHZ state is $\mathcal{D}(\ket{\psi},\ket{GHZ})=\epsilon$, to pass the verification test when an honest verifier is in the presence of dishonest agents who can perform local unitaries and communicate with each other is $P(\ket{\psi})\leq 1-\epsilon^2/4$.
%\begin{equation}
%    P(\ket{\psi})\leq 1-\frac{\epsilon^2}{4}.
%\end{equation}
The main idea of our practical e-voting protocol is that the states produced by the source and potentially manipulated by the dishonest agents will be verified a large number of times in order to ensure that the state that will be eventually used for the voting part will be very close to the GHZ state. Then we will prove that states close to the GHZ state offer almost perfect privacy for the e-voting scheme.

%The only thing the dishonest parties could try to do in order to cheat is to corrupt the source, however they cannot use a state which is very far from the GHZ state since then the verification test will fail with high probability.

We start by proving the following theorem that in high level states that with high probability if the verification procedure succeeds, then the state used for the e-voting part must be close to the GHZ state:

\begin{theorem}\label{theo1}
Let $C_{\epsilon}$ be the event that the protocol does not abort and the state used for \textsf{Voting} is such that $F(\ket{\psi}, \ket{GHZ})\leq \sqrt{1-\epsilon^2}$, for some $\epsilon>0$. Then,
\begin{equation}
    \label{eq:thm1}
    P(C_{\epsilon})\leq e^{-\tfrac{ 2^{M}(\epsilon^2-4\delta)^2}{16N\epsilon^2}},
\end{equation}
where  $\delta$ is the threshold for the ratio of rejections over trials above which the protocol is aborted, $M$ is the number of coins the agent has to toss to choose between \textsf{Verification} and \textsf{Voting} and $N$ is the number of agents.
\end{theorem}

The proof of Theorem \ref{theo1} is provided in Appendix B. Note that the honest voters do not know how many corrupt agents there are and that if a dishonest agent is the verifier, the test always passes. %If we choose $M=\log_2 \frac{4N}{\delta(1-\sqrt{1-\epsilon^2})}$, we have $P(C_{\epsilon})\leq \delta$, for any small constant $\delta$. The number of repetitions of the verification for each voter is on average $2^M=\frac{4N}{\delta(1-\sqrt{1-\epsilon^2})}$, and scales linearly with $N$. Thus, the entire e-voting protocol requires $O(N^2)$ GHZ states. 
We can make the probability of using a state that is $\epsilon$-far in trace distance from the ideal one arbitrarily small by increasing the number of repetitions, as long as we have $\delta = (1-\alpha)\epsilon^2/4$ for an $\alpha \in (0,1)$; 
more precisely, by taking $M = \log_2\left[\frac{16N}{\alpha^2\epsilon^2}\ln{\left(\frac{1}{\eta}\right)}\right]$ we can make $P(C_{\epsilon})\leq \eta$ for any small parameter $\eta>0$.
Moreover, we see that for the same choices of $\delta$ and $M$, we also have the property that the protocol accepts with high probability states that are a bit closer to the perfect GHZ state, which is important so that the protocol will not always abort. Indeed, it is easy to see with Chernoff bounds that states that are $\epsilon \sqrt{\frac{1-\alpha}{1+\alpha}}$-away from the GHZ state have probability almost 1 to pass the verification test, and thus, a source that produces such states will be sufficient for a successful election. 

We assume for the remaining of the discussion that with high probability $F(\ket{\psi}, \ket{GHZ})\geq \sqrt{1-\epsilon^2}$. In this case, we prove that for each round of the protocol, the identity of the voting agent remains almost secret: 

\begin{theorem}\label{theo2}
At any round $\ell \in [N]$ with voting agent $k$ (who has unique index $\omega_k = \ell$), if the agents use a state $\ket{\psi}$ such that $F(\ket{\psi}, \ket{GHZ})\geq \sqrt{1-\epsilon^2}$ to perform \textsf{Voting}, then for the optimal strategy that any subset of malicious agents $\mathcal{D}$ can use to guess the identity of the voting agent $k$ correctly, we have 
\begin{equation}\label{eq:anon}
    \forall j \in W_H,\: \Pr[ \mathcal{D} \emph{ guess } j] = \begin{cases}
\tfrac{1}{H}+\Tilde{\epsilon} & \text{for $j=k$}\\
\tfrac{1-\Tilde{\epsilon}}{H} & \text{for $j \neq k$},
\end{cases}
\end{equation}
where $\Tilde{\epsilon}=\sqrt{\epsilon^2+\epsilon^4}$, $W_H$ is the set and $H$ the number of honest agents.
\end{theorem}

This theorem is simply based on Theorem 2 of \cite{unnikrishnan2019}. The difference is that, instead of a sender who anonymously chooses between \textsf{Verification} and \textsf{Anonymous Transmission}, we have a voting agent who anonymously chooses between \textsf{Verification} and \textsf{Voting}. The probabilities for the other agents come from the fact that all the agents that are not voting perform exactly the same transformation on the state, so it is impossible for the dishonest parties to distinguish between them, hence the probability of guessing their identity is the same. The proof of Theorem \ref{theo2} is provided in Appendix C.%We can summarize theorem \ref{theo2}

The last property we prove shows that if the agents are all honest and use a state close to the GHZ state for voting, then the probability there is an error in the tally is small:

\begin{theorem}\label{theo3}
If at round $\ell$ the agents are honest and use a state $\ket{\psi}$ such that $F(\ket{\psi}, \ket{GHZ})\geq \sqrt{1-\epsilon^2}$ to perform \textsf{Voting}, then the probability that there is an error in the tally in the $\ell$-th round is upper bounded by $\epsilon$,
\begin{equation}
P^{er}_{\ell}\leq \epsilon.
\end{equation}
\end{theorem}

The proof of Theorem \ref{theo3} is provided in Appendix D. The above three theorems allow us to formalize and prove a number of important properties for our e-voting scheme, namely correctness, privacy, authentication, no double voting, verifiability, and receipt freeness.
%These properties of general e-voting schemes are summarized in Appendix E.
Note that in~\cite{chevallier2010}, the authors show that there exist sets of properties that are incompatible in any voting system, meaning that not all of them can be fulfilled simultaneously by any protocol. However, one can remove the incompatibility by defining approximate
versions for the properties or making computational assumptions about the voters' behaviour. Indeed, here, given that we want to allow for imperfect sources of quantum states in order to have a practical protocol that is robust to some level of noise, we define approximate versions of some of these properties for our e-voting protocol, as we explain in the following.

%We now define and prove these properties.

$(\sigma_H,\sigma_D, \gamma)$\textit{-Correctness.}
    The correctness of a protocol implies that when no adversary interferes, the election should be carried out correctly, and that in the presence of adversaries, if the election tally is far from the real votes, then the election is rejected with high probability. These two requirements can be expressed as two  properties of the voting scheme:
    \begin{itemize}
        \item $\sigma_H$\textit{-completeness:} if all agents are honest, the election result is accepted with probability more than $\sigma_H$,
    \begin{equation}\label{eq:correctness1}
    \Pr[\text{election accepted}] \geq \sigma_H.
    \end{equation}
    \item $(\sigma_D,\gamma)$\textit{-soundness:} the probability that the election result is accepted, given that the set of the votes $\mathbf{E}$ computed from the bulletin board $\mathbf{B}$ resulting from the election is more than $\gamma$-away from the real votes $\mathbf{V}$, is smaller that $\sigma_D$,
    \begin{equation}\label{eq:correctness2}
    \Pr[\text{election accepted} \; | \; \tfrac{1}{N}||\mathbf{V} - \mathbf{E}||_1
    \geq \gamma ] \leq \sigma_D.
    \end{equation}
%where $||x_n||_1=\sum_n |x_n|$.
    \end{itemize}

The use of an imperfect state may result in some errors in the final tally (see Theorem \ref{theo3}), and this is why we define a notion of approximate correctness.
In particular, the probability that the e-voting is validated is the probability that the \textsf{LogicalOr} subroutine in Phase 3 outputs 0 despite some voters announcing a wrong entry in the tally. Note that Theorem \ref{theo3} ensures that at each round we can have an error with probability at most $\epsilon$, while it can also be proven (see Lemma $\ref{Lemma2}$ in Appendix A) that during \textsf{LogicalOr}, if $j$ agents input 1 (which corresponds to their vote being tallied wrongly) the probability that \textsf{LogicalOr} outputs 0 is $S^{j}$, for the parameter $S$ defined by the e-voting protocol. Summing over all the combinations we get:
\begin{eqnarray*}
    \Pr[\text{election accepted}] = & \\ \footnotesize{ \sum_{j=0}^N \Pr[j \mbox{ inputs } 1]\Pr[\mbox{ \textsf{LogicalOr} outputs } 0 | j \mbox{ inputs } 1] } = & \\
     \sum_{j=0}^N {N \choose j} \epsilon^j (1-\epsilon)^{N-j}S^{ j}= \left[1-\epsilon(1-S)\right]^N. %&\geq \sigma_H
\end{eqnarray*}

We can then define the $\sigma_H$ parameter for our e-voting protocol as
\begin{equation}\label{conditionCorrectness1}
    \sigma_H = \left[1-\epsilon(1-S)\right]^N,
\end{equation}
and we can see that by choosing $S = 1-\chi/(\epsilon N)$ for some small constant $\chi$ we can make $\sigma_H$ close to 1.

Consider now the events $A=\{\text{The protocol produced more than } N\gamma \text{ errors}\}$ for $0\leq\gamma\leq 1$ and $B=\{ \text{The elections are validated} \}$. Then we have $P(B|A)\leq S^{N\gamma }$ %\leq \sigma_D
and we can define the $\sigma_D$ parameter of our protocol as
\begin{equation}
    \sigma_D=S^{N\gamma }.
\end{equation}
If we assume that $\gamma$ is a small fraction $\lambda$ greater than the expected number of errors, namely $\gamma=(1+\lambda)[\epsilon(1-\eta)+\eta]$, we can make $\sigma_D$ close to 0.

In conclusion, our e-voting protocol with inputs $S,\epsilon,\delta,\eta, N$ and for a small constant $\lambda>0$, is  $(\left[1-\epsilon(1-S)\right]^N, S^{N(1+\lambda)[\epsilon(1-\eta)+\eta] } , (1+\lambda)[\epsilon(1-\eta)+\eta])$-correct, where the first parameter tends to 1 and the second to 0 for an appropriate parameter $S$. 

$\zeta$\textit{-Privacy.} The privacy of the election scheme implies that each vote must remain secret with high probability. More precisely, with high probability, for any voter $k$, the probability that any subset of malicious parties $\mathcal{D}$ that deviates from the honest protocol can guess the vote $v_k$ of the voter is at most $\zeta$ more than in the case they just have access  to the bulletin board and to their own votes. In other words,
    \begin{equation}\label{eq:privacy}
\forall k, \;\;\;\Pr [v_k | \mathcal{D}] - \Pr [ v_k| \mathbf{B},v_j\in \mathbf{V}_D] \leq \zeta.
    \end{equation}
Theorems \ref{theo1} and \ref{theo2}  ensure that by repeating the \textsf{Verification} test a significant number of times at each voting round, the voting only happens with a shared state that is close to a GHZ state, which guarantees almost perfect anonymity. In practice, by having each agent record the frequency of failures of the test, they can deduce the practical parameter $\delta_k=\tfrac{\text{rejections}_k}{\text{trials}_k}$ and in case this is above the predetermined threshold $\delta$, which is an input of the protocol, the protocol is aborted. Otherwise, the rounds proceed normally and all agents vote. Note that $\delta$ is linked to the expected fidelity of the state produced by the GHZ source, as explained earlier.

We have also seen that by taking the appropriate parameters, we can have that with probability at least $(1-\eta)$, Eq.~(\ref{eq:anon}) from Theorem \ref{theo2} holds. 
%\begin{equation}
%    \forall j \in W_H,\: \Pr[ \mathcal{D} \emph{ guess } j] = \begin{cases}
%(\tfrac{1}{H}+\epsilon) & \text{for $j=k$;}\\
%\tfrac{1-\epsilon}{H} & \text{for $j \neq k$;}
%\end{cases}
%\end{equation}
%where $W_H$ is the set and $H$ the number of honest agents. 
This is the case the event $C_\epsilon$ (see Theorem \ref{theo1}) is false. In case $C_\epsilon$ is true, which happens with probability at most $\eta$, we can assume that the anonymity is totally violated. 

One needs to be careful here because the definition of privacy is not the same as the one of anonymity. More specifically, anonymity ensures that the honest voters' secret indices, or the round in which they voted, remain secret, whereas privacy implies that their vote remains a secret. Of course, the violation of anonymity implies the disclosure of privacy, however a malicious agent can gather information about someone's vote also by looking at the distribution of the other votes and the anonymity of the other voters. Taking Eq.~(\ref{eq:anon}) into account and considering that among the $H$ honest voters, $H_0$ voted for candidate `0' and the others $H_1$ voted for candidate `1', such that $H=H_0+H_1$, we have that the probability of a subset of dishonest agents guessing correctly the vote of agent $k$ that is `0' (same for `1') is the probability that they can guess that agent $k$ is part of the subset $H_0$, in other words that agent $k$ voted in one of the rounds where the vote was cast as `0'. Theorem \ref{theo2} tells us how much the dishonest agents can guess if a particular agent was the voter in a particular round, depending on whether the agent was actually the voter or not. Hence, assuming that the event $C_\epsilon$ does not hold for any round, we have
\begin{align*}
    \Pr[\mathcal{D} \mbox{ guesses } v_k=0]=
    &\tfrac{1}{H}+\Tilde{\epsilon} +(H_0-1)\tfrac{1-\Tilde{\epsilon}}{H}=\\
    =\tfrac{H_0}{H}+\tfrac{H+1-H_0}{H}\Tilde{\epsilon}&\leq  P[v_k=0|\mathbf{B}]+\Tilde{\epsilon},
\end{align*}
where we used the fact that $\tfrac{H_0}{H}$ is the distribution of the votes given by the public bulletin board and that $H_0\geq 1$.
Given that the event $C_\epsilon$ happens for each round with probability at most $\eta$ we have that for the final privacy, 
\begin{align*}
    \Pr[\mathcal{D} \mbox{ guesses } v_k=0] \leq\\ \leq P[v_k=0|\mathbf{B}]+ (1-\eta)^N \Tilde{\epsilon} + (1-(1-\eta)^N),
\end{align*}
which proves Eq.~(\ref{eq:privacy}) in the non-ideal case.

In conclusion, our e-voting protocol with inputs $S,\epsilon, \delta, \eta, N$ is $ \zeta$-private with $\zeta=(1-\eta)^N \epsilon\sqrt{1+\epsilon^2} + (1-(1-\eta)^N)$, which tends to 0 for small enough $\eta$ and $\epsilon$.

\textit{Authentication.}
Only eligible voters are allowed to vote. Our e-voting protocol as described here does not provide authentication, which should be taken care by the physical implementation of the protocol. For electronic voting machines authentication might be provided by an official ID, whereas for voting directly through the internet authentication would require some digital signature scheme. %that provide post-quantum security.

\textit{Double voting.}
Each voter can vote at most once.
Double voting is taken care of easily if the number of voters is known in advance, which in fact is necessary in our scheme in order to prepare the shared quantum state. If $N$ agents declare they want to vote, we will have an $N\times N$ bulletin board, each row of which corresponds to one vote.  A null vote can be treated as an additional candidate and will be discussed below. A dishonest voter might try to intercept all the transcripts, modify the bulletin board by adding a column and a row with another vote without changing the sum of each row, but this would result in a evident $(N+1)\times(N+1)$ matrix that will be rejected by the honest voters. At the same time, if a dishonest voter keeps the same number of rows and columns but tries to vote at a round where they are not supposed to be voting, then either the vote will not change or if the vote in the ballot changes from the intended vote of the honest voting agent, then this will be captured by the \textsf{LogicalOr} subroutine protocol in Phase 3 of the protocol.

\textit{Verifiability.} Each voter can verify that their vote has been counted correctly. More precisely, a protocol is called verifiable if there exists a function $g$ specified by the protocol, such that every voter can apply the function $g$ on the bulletin board and a  private witness $w_k$ (the witness corresponds to the vote and the secret voting index of the voter) and get back 1 if and only if their vote was counted correctly. In other words,
    \begin{align}\label{eq:verifiability}
        \exists  &g \; \mbox{s.t.} \;  \forall k \; \exists w_k \; \mbox{s.t.} \; \forall \mathbf{B}: \;\;\; \\
    \nonumber    g(\mathbf{B}, w_k)= &1 \: \iff v_k \;\mbox{ was counted in the tally}
    \end{align}
The verifiability, thus, demands the existence of a function that, given the bulletin board $\mathbf{B}$ and the voter's secret index $\omega_k$ returns $1$ if $v_k$ was counted in the tally and $0$ otherwise.

The verifiability is inherent in the protocol, as the tally is performed by the voters themselves. The bulletin board produced as an output of the protocol is public and can always be checked by everyone, however it appears as a random set of votes. Each row $j$ corresponds to the vote $v_k$ of agent $k$ whose secret unique index is $\omega_k=j$, and thus each agent can easily verify their own vote and only that one. If the vote in the bulletin board differs from the actual intended vote $v_k$, as a consequence of a dishonest behaviour or an imperfection in the quantum state, the agent can reject the result through the \textsf{LogicalOr} subroutine in Phase 3.

\textit{Receipt freeness.} A voter cannot prove how they voted, in order to avoid vote selling. A receipt is a witness $w_k$  defined as:
    \begin{equation}
        \exists  g \; \text{s.t.} \; \forall k \; \exists  v_k \; \exists ! w_k \; \text{s.t.} \; \forall \mathbf{B} \;\; g(\mathbf{B},k ,v_k, w_k)=1
    \end{equation}
    If there is no receipt, then the protocol is called receipt-free.
As long as their index stays secret all the agents can always deterministically verify their votes, without getting their privacy violated, and without producing any receipt of their vote which could be used for vote selling.

%%%%%%%%%%%%%%%%%%%%%%%%%%%%%%%%%

%\section*{Additional candidates}

\emph{Additional candidates.}--So far we assumed that there were only two candidates, `0' or `1', which is suitable for referendum type of elections. 

We can easily deal with the case of more than two candidates by repeating the  e-voting protocol multiple times in sequence. In particular, if there are $K$ candidates, we can express each of them with a binary number of $\log_2 K$ digits and repeat the election as many times, so that each vote set $\mathbf{E}$ corresponds to one digit of the candidates (see Appendix E for more details).
The only properties that are affected by the additional candidates are correctness and privacy. This is because they are the only ones that are probabilistic and that actually depend on the use of an imperfect state in the different rounds. Authentication, double voting, verifiability and receipt freeness will thus remain unchanged even in the scenario with many candidates. In particular, the correctness is affected because repeating the elections multiple times increases the probability of having an error at some point. We can assume that the agents perform the \textsf{LogicalOr} protocol with security parameter $S$ to notify an error only at the end of all the repetitions of the elections. By Theorem 3, the probability that for any agent at least one bit of the final tally will be incorrect is $\epsilon^*=1-(1-\epsilon)^{\log_2 K}$ and thus the probability that the election is accepted after multiple rounds will be
\begin{equation}
    \text{Pr}[\text{election accepted}]=[1-\epsilon^*(1-S)]^N=\sigma_H^*.
\end{equation}
The soundness, on the other hand, is not affected. In fact, if any voter notices more than one incorrect bit in their vote it will count as a single error. 

The privacy of the protocol is affected as well. From one point of view, since the number of bits that the dishonest agents need to guess is larger, the probability of violating the anonymity, and thus the privacy, is actually smaller. If, however, we want that each individual bit of the vote remains private, the privacy is decreased by the fact of having multiple rounds of elections. In this case, let us consider the probability of the event $X$ that at some round the malicious agent guesses the preference of the voter $k$ and let us assume that this probability is upper bounded as $P(X)\leq \zeta$. If we repeat the elections $\log_2 K$ times, the probability that event $X$ is true at least once will thus be at most $\zeta^*=1-(1-\zeta)^{\log_2K}$. Hence, when dealing with multiple candidates our e-voting protocol with inputs $S,\epsilon, \delta, \eta, N$ is $ \zeta^*$-private with $\zeta^*=1-(1-(1-\eta)^N \epsilon + [1-(1-\eta)^N)]^{\log_2 K}$, which tends to 0 for small enough $\eta$ and $\epsilon$. 

%%%%%%%%%%%%%%%%%%%%%%%%%%%%%%%%%

\section*{Example}
 Let us consider a 4-photon GHZ source that produces states with fidelity $\sim 0.85$, which is realistic with state-of-the-art quantum photonic technology.
This corresponds to an expected fraction of rejections $\delta \sim 0.05$. If we now set $\epsilon = 0.6$, it implies that we will never accept states with fidelity lower than $\sim 0.8$. Although we would like $\epsilon$ to be very small, it is not possible to reach values lower than the above with present technology. In any case, this is not an issue for the correctness
of the protocol. If we fix the non aborting probability $\eta = 0.001$, Eq. (12) implies that $M=12$, and thus we will need around $4000$ GHZ states to accomplish each voting. Since the production rate of currently available GHZ sources is $\sim 8 kHz$, the protocol can be carried out in just a few seconds. The actual problem is privacy, which would be violated
with a probability of $\zeta \sim 0.7$. A way to tackle this issue is to amplify privacy by repeating $Q$ rounds of the e-voting protocol and encoding the vote intention of each voter in the parity of the $Q$ outcomes. In this way, each round of the election would encode no information and the malicious agents would require to succeed in each of the rounds, reducing the
total probability of violating privacy to $\Tilde{\zeta} = \zeta^ Q$. If $Q = 15$, with the previous parameters this would reduce the privacy parameter to $\Tilde{\zeta} \sim 0.005$, at the expense of repeating more rounds of the elections. Notice that this privacy amplification procedure would have the consequence of increasing the potential number of errors in the outcome of the election. This and other details of the protocol are left for consideration in future experimental implementations.

\section{Conclusion} We have described and analyzed a practical quantum e-voting scheme and provided approximate definitions of correctness and privacy, which make it appropriate for realistic non-ideal scenarios. The quantum e-voting protocol that we have described achieves information-theoretic security without requiring trust in the quantum source or in any election authority. Previously proposed classical schemes, such as the one in Ref.~\cite{broadbent2007}, also achieve information-theoretic security, however the requirement of trusting authorities and simultaneous broadcasting could make it impractical. A small-scale election demonstration of our protocol can be implemented with currently available quantum photonic platforms and with the improvement of these technologies a voting scheme for board meetings and similar scenarios may be attainable in the near future. When GHZ states of several photons, with high fidelity and a reasonable repetition rate, %, for example through a photonic quantum computer,
become available and can be well controlled, it will be possible to implement the protocol at a metropolitan level and then as a consequence, with a subdivision into regular elections, at a national level. Although this is certainly challenging, all the future applications of quantum information protocols will have to meet similar obstacles and a future realization of this protocol might be an impactful practical use case of quantum technologies.

\section*{Acknowledgments} We acknowledge financial support from the European Research Council project QUSCO (E.D.) and the French National Research Agency project quBIC.

\bibliographystyle{ieeetr}
\bibliography{evoting_bib}

%\begin{thebibliography}{1}

%\end{thebibliography}

\newpage

\appendix
\section{Appendix A: Subroutines}\label{appendixSubroutines}
We provide the details for each subroutine used in the e-voting protocol. Before that we lay out some useful notation.

\begin{itemize}
    \item $N$: the total number of participants to the election;
    \item $\mathbf{W} = \{1,2,\ldots,N \}$: the set of all the voters. $\mathbf{W_H}$ and $\mathbf{W_D}$ the sets of honest and dishonest voters respectively;
    \item $\mathbf{V}=\{v_k\}_{k\in W}$: the set of \textit{votes}. Each voter $v_k$'s value is the index of the candidate for which they want to vote;
    \item $K$: the number of eligible candidates;
    \item $\mathbf{C}= \{ 0,1,\ldots,K-1 \}$: the set of candidates. We first assume $\mathbf{C}=\{0,1\}$; the generalization to more candidates is shown in the dedicated section.
    \item $\mathbf{B}=\{b^j_k\}$: the bulletin board encodes all the anonymous votes to be tallied.
    \item $\mathbf{E}=\sum_k \mathbf{B}$ is the set of votes resulting by summing the rows of the bulletin board. Errors or dishonest players may induce some $b_k$ to be different from $v_k$.
    \item $\mathbf{T}=\{t_i\}_{i\in C}$  the tally, a vector whose elements represent the number of votes for the corresponding candidate. It can always be computed as a vector valued function of the bulletin board $f(\mathbf{B})=\mathbf{T}$.
    \item $\mathbf{R}=\{r_i\}_{i\in \mathbf{C}}$ is the result of the elections with the actual set of voters $\mathbf{V}$. It is the histogram of the real voters preferences.
\end{itemize}

Let un now take a look at the specific subroutines employed in the \textsf{Quantum e-voting} protocol. The \textsf{LogicalOr}, \textsf{RandomBit} and \textsf{RandomAgent} subroutines are classical anonymous protocols taken from \cite{broadbent2007} and used in \cite{unnikrishnan2019}. In particular, the last two are based on the first one, which performs the logical OR of all the agents' inputs. It will thus output 1 with high probability if and only if at least one agent had input 1. The \textsf{RandomBit} subroutine employs the \textsf{LogicalOr} to produce shared randomness, \emph{i.e.}, a random bit publicly announced according to some probability distribution. This can be used a number of times in order to draw an agent at random among the voters through \textsf{RandomAgent}. %Finally, \textsf{Notification} is simply the \textsf{LogicalOr} subroutine where each agent that wants to anonymously notify the others inputs 1, so that either the output is that no agent wants to notify anything or at least one agent wants to notify the others.

\begin{algorithm}[]
\caption{\textsf{LogicalOr}}
\begin{flushleft}
\textit{Input:} $N$ agents, $N$ boolean variables ${x_i}$, security parameter $S=(1-2^{-\Gamma})^{\Sigma}\in (0,1)$. \\
\textit{Output: } $y=\bigvee_i^N x_i$.\\
\textit{Resources: } Classical communication and random numbers.\\
\textit{Description: }
\end{flushleft}
\begin{algorithmic}[1]
    \STATE Decide $N$ random orderings, such that each voter is the last once. For each ordering repeat $\Sigma$ times the following.\\
    \STATE Each voter $k$ gives an input $x_k$.\\
    \STATE If $x_k=0$ set $p_k=0$, otherwise toss $\Gamma$ coins and set $p_k$ to 1 if the result is `all heads' and to 0 otherwise.\\
    \STATE Then each voter generates uniformly at random an $N$-bit  string $r_k=r_k^1 r_k^2... r_k^N$, such that $\bigoplus_{i=1}^N  r_k^i=p_k$. \label{step3}\\
    \STATE Voter $k$ sends $r_k^i$ to voter $i$ for all $i$, keeping $r_k^k$ for themselves.\\
    \STATE Each voter sums the received bits and broadcasts the parity $z_i=\bigoplus_{k=1}^N r_k^i$ according to the ordering.\\
    \STATE Compute the parity of the original bits $y=\bigoplus_i z_i$.\\
    \STATE From this everyone can also compute the parity of all other inputs except their own $w_k=\bigoplus_{i=1}^N (z_i\otimes r_k^i)$.\\
    \STATE Repeat $\Sigma$ times from step \ref{step3}: each time  repeat with $p_k$ as new inputs.\\
    \STATE If at least once in the $\Sigma$ repetitions for the various orderings $y=1$, this is the output of the protocol, otherwise it is $y=0$.
\end{algorithmic}
\end{algorithm}

The \textsf{LogicalOr} functionality is implemented probabilistically by assigning a random value $p_k$ to all inputs $x_k=1$, while $p_k=0$ if $x_k=0$. Then the parity of the $p_k$ is computed anonymously for  various orderings, such that each voter is last once, and for  repetitions for each ordering. Since the inputs of the parity are random, if at least one voter has input 1, the output of the parity will be 1 at least once through all the repetitions. The orderings are necessary for the voters to broadcast their computation asynchronously, while at the same time avoiding that the last agent changes their output to corrupt the result. This subroutine has two additional parameters as input $\Sigma$ and $\Gamma$ that in turn define the security parameter $S$. $\Sigma$ indicates the number of times the protocol needs to be repeated for each ordering, while $\Gamma$ specifies the number of coins that each voter has to toss to assign the value $p_k$, which will be 1 only if the result is `all heads'. As a consequence, the security parameter $S=(1-2^{-\Gamma})^{\Sigma}$ can take any value in the open interval $(0,1)$ and represents the probability of the protocol giving the incorrect answer.

The following lemmas are taken from Ref.~\cite{broadbent2007}.
\begin{lemma}
(Reliability) No one can abort the \textsf{LogicalOr} protocol.
\end{lemma}
If someone refuses to broadcast, it is assumed that the output of the protocol is 1.
\begin{lemma}\label{Lemma2}
(Correctness) If all the inputs are $x_i=0$, the \textsf{LogicalOr}   protocol outputs $y=0$ with probability 1. If $N$ agents input 1 in the protocol then we will have $y=1$ with probability at least $P=1-S^{N}$.
\end{lemma}
\begin{lemma}
(Privacy) The most an adversary can know in the protocol is the logical Or of the other participants.
\end{lemma}
These properties are also guaranteed  in the following subroutines that are based on \textsf{LogicalOr}.

\begin{algorithm}[]
\caption{\textsf{RandomBit}}
\begin{flushleft}
\textit{Input:} Security parameter $S$ to be used in \textsf{LogicalOr}. \\
\textit{Output: } The voting agent anonymously announces a bit uniformly at random.\\
\textit{Resources: } Classical communication and random numbers.\\
\textit{Description: } Perform the \textsf{LogicalOr} with security parameter $S$ where the voting agent inputs a random bit according to $D$ and the other agents input 0.
\end{flushleft}
\end{algorithm}

\begin{algorithm}[]
\caption{\textsf{RandomAgent}}
\begin{flushleft}
\textit{Input:} Security parameter $S$ to be used in \textsf{RandomBit}.\\
\textit{Output: } The voting agent anonymously chooses an agent uniformly at random.\\
\textit{Resources: } Classical communication and random numbers.\\
\textit{Description: } Repeat \textsf{RandomBit} $\log_2N$ times.
\end{flushleft}
\end{algorithm}

\iffalse

\begin{algorithm}[]
\caption{\textsf{Notification}}
\begin{flushleft}
\textit{Input:} security parameter $S$, \textit{voting agents:} boolean input $x$\\
\textit{Output: } the agents anonymously are notified all the agents if th bit $x$.\\
\textit{Resources: } Classical communication and random numbers.\\
\textit{Description: } Perform the \textsf{LogicalOr} where the \textit{sender} inputs $x$ and the other agents input 0.
\end{flushleft}
\end{algorithm}

\fi

\textsf{UniqueIndex} is used to anonymously distribute a secret random index to each voter. Note that here it is a classical protocol while in \cite{wang2016} it was necessary to use another entangled quantum state to achieve the same goal. This protocol is polynomial in the number of the operations and completely guarantees the privacy.
In order to achieve this functionality we proceed in the following way. The protocol is composed of $N$ rounds. In the first step of each round all agents perform the \textsf{LogicalOr} protocol with inputs 0 if they already have an index, otherwise they will input 1 with probability $1/t$ and 0 with probability $1-1/t$, where $t$ is the number of agents that do not have an index yet. If there is any agent with input 1 the output of \textsf{LogicalOr} will be $y=1$. Each agent with input $x_k=1$ can verify at this point if there is a collision by tracking the parity of all other inputs $w_k$. If for any of the $\Sigma$ repetitions in every ordering $w_k\neq 0$, then they know that there is someone else with input 1. At the end of each \textsf{LogicalOr} everybody performs another \textsf{LogicalOr} protocol that acts as an anonymous notification, in which they input 0, unless no collision was detected. Everyone then repeats the first \textsf{LogicalOr}; this time those who previously had input 0 will stay the same, while the others toss a coin and decide their inputs accordingly. This is repeated until there is only one agent $j$ with input 1, while $w_j=0$ throughout all repetitions of \textsf{LogicalOr}. When the notification \textsf{LogicalOr} is performed, agent $j$ will be the only with input 1, announcing that the index $\omega_j$ was assigned and the round is over. Then this is repeated from the first step, the agents who already have an index always set their input to 0 and the protocol terminates when the last notification \textsf{LogicalOr} output is 0, announcing that all indices have been assigned. If at any time $y=0$, then there is no one with input 1, and the protocol should be repeated from the beginning of the last \textsf{LogicalOr}, with the same inputs until someone gets an index.

\begin{algorithm}[h]
\caption{\textsf{UniqueIndex}}
\begin{flushleft}
\textit{Input:} Security parameter $S$ to be used in \textsf{LogicalOr}, $N$ random boolean variables ${x_i}$.\\
\textit{Output: } Each agent $k$ has a secret unique index $\omega_k$.\\
\textit{Resources: } Classical communication and random numbers.\\
\textit{Description: }
\end{flushleft}
\begin{algorithmic}[1]
    \STATE Beginning of round $R=1$.\\
    \STATE Perform \textsf{LogicalOr} with inputs $x_k=0$ if they already have an index, otherwise they input $x_k=0$ with probability $1-1/(N-R)$ and $x_k=1$ with probability $1/(N-R)$.\\
    \STATE If $y=0$ repeat from step 2.\\
    \STATE If an agent $k$ has a bit $x_k=1$ and $w_k=0$ they know they are the only one and has been assigned the secret index corresponding to the round $\omega_k=R$, otherwise there is a collision.\\
    \STATE [notification] Everybody performs a \textsf{LogicalOr} with input 0, unless they received the index in this round, in which case they input 1.\\
    \STATE If the output of \textsf{LogicalOr} is 0, no index was assigned and we repeat from step 2.\\
    \STATE If the output of \textsf{LogicalOr} is 1, the index was assigned and we repeat from step 2 with $R\rightarrow R+1$.
    \STATE Repeat from step 2 until all indices have been assigned.
\end{algorithmic}
\end{algorithm}
%Where $w_k$ (not to be confused with $\omega_k$) was defined in the \textsf{LogicalOr} protocol and represents the parity of all the inputs except the one with index $k$.

\textsf{Verification} is the same protocol as in \cite{pappa2012}, where a test is performed by all the agents and the quantum state will pass it with a probability that grows with the fidelity between the input state and an ideal GHZ state.

\begin{algorithm}[H]
\caption{\textsf{Verification}}
\begin{flushleft}
\textit{Input:} A quantum state distributed and shared by $N$ parties, security parameter $S$ for \textsf{RandomAgent}. \\
\textit{Output: } If the state is a GHZ state $\rightarrow$ YES.\\
\textit{Resources: }  Classical communication, random numbers, quantum state source, quantum channels.\\
\textit{Description: }
\end{flushleft}
\begin{algorithmic}[1]
    \STATE Everyone executes \textsf{RandomAgent} to choose uniformly at random one of the voters to be the \textit{verifier}.\\
    \STATE The verifier generates random angles $\theta_j\in [0, \pi)$ for all agents including themselves, such that the sum is a multiple of $\pi$. The angles are then sent out to all the agents.\\
    \STATE Agent $j$ measures in the basis $\left[\ket{+_{\theta_j}},\ket{-_{\theta_j}}\right]=\left[ \frac{1}{\sqrt{2}} \left( \ket{0}+e^{i \theta_j}\ket{1} \right)  , \frac{1}{\sqrt{2}} \left( \ket{0}-e^{i \theta_j}\ket{1} \right)  \right]$
    and publicly broadcasts the result $Y_j=\{0,1\}$.\\
    \STATE The state passes the verification test when the following condition is satisfied: if the sum of the randomly chosen angles is an even multiple of $\pi$, there must be an even number of 1 outcomes for $Y_j$, and if the sum is an odd multiple of $\pi$, there must be an odd number of 1 outcomes for $Y_j$ : $ \bigoplus_j Y_j = \frac{1}{\pi} \sum_i \theta_i $.
\end{algorithmic}
\end{algorithm}

With \textsf{Voting} a voter can express their preferred candidate. The state that will be used for voting is equivalent to a GHZ state up to a local Hadamard transform applied by each agent to their own particle. Once the GHZ state is measured in the Hadamard basis, the outcomes will always sum up to $0 \mod 2$. This can be seen by direct application of the $N$-dimensional Hadamard $\mathcal{H}^{\otimes^N}=\mathcal{H}_1\otimes\mathcal{H}_2\otimes...\otimes\mathcal{H}_N$, where each of the transforms $\mathcal{H}_j$ acting on the 2-dimensional Hilbert space of the $j$-th voter's particle is expressed in the computational basis as
\[
\mathcal{H}_j=\frac{1}{\sqrt{2}}\left[(\ket{0}_j+\ket{1}_j)\bra{0}_j + (\ket{0}_j-\ket{1}_j)\bra{1}_j \right].
\]
It is easy to show that if we apply the Hadamard to the GHZ state we obtain:

{\footnotesize
\begin{align*}
     \mathcal{H}^{\otimes^N} \ket{GHZ}=2^{-N}\left[\bigotimes_{i=1}^N\left(\ket{0}_i+\ket{1}_i\right)+\bigotimes_{i=1}^N\left(\ket{0}_i-\ket{1}_i\right)\right]=\\
     =2^{-N}\left[\sum_{\{ k_i = 0,1 \}_{i=1}^N} \ket{k_i}^{\otimes_i^N}+\sum_{\{ k_i = 0,1 \}_{i=1}^N} (-1)^{\sum k_i}\ket{k_i}^{\otimes_i^N} \right]=\\
     =2^{-N/2}\sum_{\sum k_i = 0_{\text{mod} 2}} \ket{k_i}^{\otimes_i^N}.
\end{align*}
}

So, by measuring each particle in the Hadamard basis, we are assured that the sum of the outcomes will be 0 modulo 2.

\begin{algorithm}[H]
\caption{\textsf{Voting}}
\begin{flushleft}
\textit{Input:} Voting agent preference $v_k$. \\
\textit{Output: } All agents get one row of the bulletin board.\\
\textit{Resources: }  Classical communication, GHZ source, quantum channels.\\
\textit{Description: }
\end{flushleft}
\begin{algorithmic}[1]
    \STATE Each agent measures the state they received in the Hadamard basis and records the outcome.\\
    \STATE The outcomes of the measurement of each voter $k$ is $d_k$. Then we know that $\sum_k d_k = 0 \mod 2$.\\
    \STATE The voting agent performs an XOR between the outcome $d_k$ and their vote $v_k$:  $d_k\rightarrow B_k= d_k\oplus v_k$. However, this alone will still appear as a random string.\\
    \STATE Every agent publicly broadcasts $d_k$  which gives one line $\mathbf{b_k}$ of the bulletin board  $\mathbf{B}=\{\mathbf{b_k}\}$.
\end{algorithmic}
\end{algorithm}

\section{Appendix B: Proof of Theorem \ref{theo1}} \label{appendixTheorem1}

Here we prove the soundness of the \textsf{Verification} protocol. For simplicity of the proof, recall that we denote the ideal state by $\ket{\Phi_0^n}$,
which can be obtained from the GHZ state by applying a Hadamard and a phase shift $\sqrt{Z}$ to each qubit.
\begin{manualtheorem}{1}
Let $C_{\epsilon}$ be the event that the protocol does not abort and the state used for \textsf{Voting} is such that $F(\ket{\psi}, \ket{GHZ})\leq \sqrt{1-\epsilon^2}$, for some $\epsilon>0$. Then,
\begin{equation}
    \label{eq:thm1}
    P(C_{\epsilon})\leq e^{-\tfrac{ 2^{M}(\epsilon^2-4\delta)^2}{16N\epsilon^2}},
\end{equation}
where  $\delta$ is the threshold for the ratio of rejections over trials above which the protocol is aborted, $M$ is the number of coins the agent has to toss to choose between \textsf{Verification} and \textsf{Voting} and $N$ is the number of agents.
\end{manualtheorem}

\begin{proof}

%Our aim is to bound the probability that the protocol does not abort and the fidelity of the state $\ket{\Psi}$ used for \textsf{Voting} is given by $F'(\ket{\Psi}) = \underset{U}{\max \ } F(U \ket{\Psi}, \ket{\Phi_0^n}) \leq \sqrt{1 - \epsilon^2}$, where $U$ is a general operator on the space of the malicious agents.

During the protocol, each voter can trust only themselves as they do not know who could be a colluding agent. Thus, although at each round of \textsf{Verification} a verifier is chosen at random and could be an honest voter, we will perform the following analysis assuming we are in the worst case scenario in which the voting agent is the only honest voter and cannot trust anybody else. Thus the average number of rounds of the \textsf{Verification} will be $\langle D \rangle=2^M/N$. In addition, if we take $M$ large enough, we can make the probability of having at least $D=2^M/2N$ rounds of \textsf{Verification}, arbitrarily close to 1. %Conversely, we can assume that the \textit{voting agent} will proceed with \textsf{Voting} only after a minimum number of rounds, so that the dishonest agents cannot exploit the extremely rare case in which the state is used for \textsf{Voting} without having been verified for  enough rounds. 
Thus, in the following we will assume that $D\geq 2^{M}/2N$. In any practical implementation of the protocol, however, the other honest agents will also assist the verification and if they count a ratio of rejections larger than $\delta$ they can abort the elections, increasing the soundness of the protocol.

Although we allow the malicious source to create any state in any round and even entangle the states between rounds, the optimal cheating strategy, which maximizes the probability of the event $C_\epsilon$, is to create in each round some pure state $\ket{\Psi}$ such that $F'(\ket{\Psi}) = \sqrt{1-\epsilon^2}$, as proven in \cite{pappa2012}. In high level, one can first see that an entangled strategy does not help, as it can be replaced by a strategy sending unentangled states as follows. Given some entangled state, for a given round, the probability of passing the test and the fidelity of the state depend only on the reduced state, conditioned on passing previous rounds. The same effect is achieved by sending these mixed reduced states corresponding to each round, without any entanglement.

Next, one sees that by providing a mixed state, the source does not gain any advantage, as a mixed state is a probabilistic mixture of pure states, and the overall cheating probability of this mixed strategy is just a weighted combination of the cheating probabilities of each of the pure states. Then, obviously this mixed strategy is worse than the strategy that always sends the pure state that has the maximum cheating probability of all states in the mixture. Hence, one can continue the proof by only considering strategies with pure states.

Moreover, since the adversary is just trying to maximize the probability the state $\ket{\Psi}$ used for voting has $F'(\ket{\Psi})= \sqrt{1 - \epsilon^2}$, it is clear that there is no need to send any state with even smaller $F'(\ket{\Psi})$, since then the probability of failing the test (and therefore the protocol aborting) would just increase. Last, if in any round the source created a state with higher $F'(\ket{\Psi})$, then this certainly does not contribute to the event $C_\epsilon$, and in fact it may also cause the protocol to abort. Thus, to upper-bound the probability of event $C_\epsilon$ with respect to the best attack a malicious source can perform, we only need to consider the case where in each round the malicious source creates some state $\ket{\Psi}$ such that $F'(\ket{\Psi}) = \sqrt{1-\epsilon^2}$.

%Before the protocol the agents must choose an $\epsilon$ such that if the source produces a state with a larger trace distance from a GHZ it will be rejected with high probability, so that we can safely assume that the states used in the protocol are at most $\epsilon$-far from the GHZ. In addition, 
The protocol takes as input a threshold parameter $\delta$, such that if during their round the voting agent rejects the state more than a $\delta$ fraction, %of the total number of times she is the verifier, 
then they abort the elections because the source is corrupted. In the limit, the ratio of rejections will tend to the probability of a single state $\epsilon$-far in trace distance from a GHZ to fail the \textsf{Verification} test in the presence of dishonest adversaries, which is \cite{pappa2012}:
\begin{equation}
    P(\epsilon)\geq \frac{\epsilon^{ 2}}{4}.
\end{equation}

%This implies that  only states with trace distance smaller than $\epsilon^{\prime}=2\sqrt{\delta}$ will pass this test with high probability. We can thus assume that the GHZ source has been previously calibrated to produce states with fidelity $F'(\ket{\Psi})\geq\sqrt{1-\epsilon^{\prime\: 2}}$. This can be obtained by performing some rounds of \textsf{Verification} before the \textsf{e-voting} and can be done publicly without compromising the privacy of the voters or the correctness of the protocol.
 
Thus, we can use a Chernoff inequality to bound the probability that in $D$ rounds of \textsf{Verification} with a state $\epsilon$-far the ratio of rejections of the voting agent $\delta_{k}$ is smaller than $\delta$, in which case the event $C_{\epsilon}$ is true. In particular, given that the expected number of rejections is at least $D \epsilon^2/4$, the Chernoff bound gives the following inequality

\begin{equation}\label{eq:chernoff}
    P(C_{\epsilon})=P(\delta_{k} \leq \delta)\leq e^{-\tfrac{ D(\epsilon^2-4\delta)^2 }{8\epsilon^2}}.
\end{equation}

%Where $\delta_{\omega_k}$ is the total number of time the \textit{voting agent} rejected a state divided by the total number of times they performed \textsf{Verification}. 

If we substitute $D\geq 2^{M-1}/N$, we obtain the expression of Theorem \ref{theo1}.

%\iordanis{LEt's try this again. I think the theorem should have two statements. One that $P[C_\epsilon]$ is small if i pick the correct delta and then that $P[C_{\epsilon/2}]$ (or $P[C_{\epsilon/100}]$) is big so I can accept states $\epsilon/2$ far from GHZ. This will give us a delta as a function of epsilon. In the end there should be no free parameter delta, but picking the right delta, now that you know the expressions, should make $P[C_\epsilon]$ go to 0 and $P[C_\epsilon]$ go to 1. Does this make sense?}

%\fede{Notice that this bound has a mathematical sense only if we consider $\epsilon> 2\sqrt{\delta}$. Equation \ref{eq:chernoff} tells us that if we use states that have  fidelity close to the one we expect from our source, the probability of using the state in the protocol is large. Conversely, when the state is far from $2\sqrt{\delta}$ in trace distance, than we can make the probability of accepting the state arbitrarily small by increasing the number of coins $M$. For instance, we can assume that the  source produces  GHZ states with fidelity $\sqrt{1-\epsilon^2/4}$  and thus $\delta=\frac{\epsilon}{16}$. In this case, we will always accept states with trace distance $\frac{\epsilon}{2}$ or smaller, while if we take $M = \log_2\left[\frac{16N}{(\epsilon^2-4\delta)^2\epsilon}\ln{\left(\frac{1}{\eta}\right)}\right]$, we can make the probability  of using a state $\epsilon$-far for the elections $P(C_{\epsilon})\leq \eta$ arbitrarily small.}

\end{proof}

\section{Appendix C: Proof of Theorem \ref{theo2}} \label{appendixTheorem2}
Next, we prove the anonymity of the protocol as in \cite{unnikrishnan2019}. Once again, recall that we denote the ideal state by $\ket{\Phi_0^n}$,
which can be obtained from the GHZ state by applying a Hadamard and a phase shift $\sqrt{Z}$ to each qubit. The voter's transformation now becomes $\sigma_x \sigma_z$.
Further, we also define the state:
\begin{align}
 \ket{\Phi_1^n}  = \frac{1}{\sqrt{2^{n-1}}} \Big[ \underset{\Delta(y) = 1 \text{ (mod 4)}}{\sum} \ket{y} - \underset{\Delta(y) = 3 \text{ (mod 4)}}{\sum} \ket{y} \Big],
 \end{align}
and note that $\sigma_x \sigma_z \ket{\Phi_0^n} = \ket{\Phi_1^n}, \sigma_x \sigma_z \ket{\Phi_1^n} = - \ket{\Phi_0^n}$.

We consider two cases here: first, when all the agents are honest (Lemma \ref{l:honest}), and second, when we have malicious agents who could apply some operation on their part of the state (Lemma \ref{l:dishonest}).
\renewcommand{\thetheorem}{2A}
\begin{lemma}
If all the agents are honest, and they share a state $\ket{\Psi}$ such that $F(\ket{\Psi}, \ket{\Phi_0^n}) = \sqrt{1 - \epsilon^2}$, then for every honest agent $i, j$ who could be the voter, we have that $F(\ket{\Psi_i}, \ket{\Psi_j}) \geq 1- \epsilon^2$, where $\ket{\Psi_i}$ is the state after agent $i$ has applied the voter's transformation.
\label{l:honest}
\end{lemma}
\begin{proof}
If we have $F(\ket{\Psi}, \ket{\Phi_0^n}) = \abs{\bra{\Psi}\ket{\Phi_0^n}}^2 = \sqrt{1 - \epsilon^2}$, then similarly to \cite{pappa2012} we can write the state shared by all the agents as:
\begin{align}
\ket{\Psi} = (1 - \epsilon^2)^{1/4} \ket{\Phi_0^n} + \epsilon_1 \ket{\Phi_1^n} + \sum_{i=2}^{2^n-1} \epsilon_i \ket{\Phi_i^n},
\end{align}
where $\sum_{i=1}^{2^n-1} \epsilon_i^2 = 1 - \sqrt{1-\epsilon^2}$. If agent $i$ is the voter, then they apply $\sigma_x \sigma_z$, and the state becomes:
\begin{align}
\ket{\Psi_i} = (1 - \epsilon^2)^{1/4} \ket{\Phi_1^n} - \epsilon_1 \ket{\Phi_0^n} + \sum_{i=2}^{2^n-1} \epsilon_i' \ket{\Phi_i^n}.
\end{align}
Instead, if agent $j$ is the voter and they apply $\sigma_x \sigma_z$, the state becomes:
\begin{align}
\ket{\Psi_j} = (1 - \epsilon^2)^{1/4} \ket{\Phi_1^n} - \epsilon_1 \ket{\Phi_0^n} + \sum_{i=2}^{2^n-1} \epsilon_i'' \ket{\Phi_i^n}.
\end{align}
The fidelity is then given by:
\begin{align}
F(\ket{\Psi_i}, \ket{\Psi_j}) & = \abs{\bra{\Psi_i}\ket{\Psi_j}}^2 \\
& = \abs{\sqrt{1 - \epsilon^2} + \epsilon_1^2 + \sum_{i=2}^{2^n-1} \epsilon_i' \epsilon_i''}^2 \\
& \geq 1 - \epsilon^2.
\end{align}
%This gives $D(\ket{\Psi_i}, \ket{\Psi_j}) = \sqrt{1 - F(\ket{\Psi_i}, \ket{\Psi_j})} \leq \epsilon$.
\end{proof}

\renewcommand{\thetheorem}{2B}
\begin{lemma}
If some of the agents are malicious, and they share a state $\ket{\Psi}$ such that $F'(\ket{\Psi}) \geq \sqrt{1-\epsilon^2}$, then for every honest agent $i, j$ who could be the voter, we have that $F(\ket{\Psi_i}, \ket{\Psi_j}) \geq 1 - \epsilon^2 $, where $\ket{\Psi_i}$ is the state after agent $i$ has applied the voter's transformation.
\label{l:dishonest}
\end{lemma}
\begin{proof}
%We have $F(U \ket{\Psi}, \ket{\Phi_0^n}) \geq \sqrt{1 - \epsilon^2}$.
Recall that our fidelity measure is given by $F'(\ket{\Psi}) = \underset{U}{\max \ } F(U\ket{\Psi}, \ket{\Phi_0^n})$. Let us now denote by $\ket{\Psi'}=U \ket{\Psi}$ the state after the operation $U$ which maximizes this fidelity has been applied. As in \cite{pappa2012}, we can write this state in the most general form as:
\begin{align}
\ket{\Psi'} = \ket{\Phi_0^k} \ket{\psi_0} + \ket{\Phi_1^k} \ket{\psi_1} + \ket{\chi},
\end{align}
where note that $\ket{\chi}$ contains both honest and malicious parts, of which the honest part is orthogonal to both $\ket{\Phi_0^k}$ and $\ket{\Phi_1^k}$.

We want to find the closeness of the states $\ket{\Psi_i}, \ket{\Psi_j}$, which are the states after the $\sigma_x \sigma_z$ operation is applied to $\ket{\Psi'}$ by either agent $i$ or $j$ who is the voter.
These states are given by:
\begin{align}
\ket{\Psi_i} & = \ket{\Phi_1^k} \ket{\psi_0} - \ket{\Phi_0^k} \ket{\psi_1} + \ket{\chi'}, \\
\ket{\Psi_j} & = \ket{\Phi_1^k} \ket{\psi_0} - \ket{\Phi_0^k} \ket{\psi_1} + \ket{\chi''}.
\end{align}
The fidelity is then given by:
\begin{align}
F(\ket{\Psi_i}, \ket{\Psi_j}) & = \abs{\bra{\Psi_i}\ket{\Psi_j}}^2 \\
& = \abs{\bra{\psi_0}\ket{\psi_0} + \bra{\psi_1}\ket{\psi_1} + \bra{\chi'}\ket{\chi''}}^2.
\end{align}
However, although the overall state $\ket{\Psi'}$ is normalized, the malicious agents' part of the state is not. Thus, we need to determine a bound on $\bra{\psi_0}\ket{\psi_0}$ and $\bra{\psi_1}\ket{\psi_1}$. We have:
\begin{align}
F( \ket{\Psi'}, \ket{\Phi_0^n}) =  \abs{\bra{\Phi_0^n}\ket{\Psi'}}^2  \geq \sqrt{1 - \epsilon^2}.
\end{align}
It was shown in \cite{pappa2012} that we can write for any $k, n$:
\begin{align}
\ket{\Phi_0^n} = \frac{1}{\sqrt{2}} \Big[ \ket{\Phi_0^k} \ket{\Phi_0^{n-k}} - \ket{\Phi_1^k} \ket{\Phi_1^{n-k}} \Big],
\end{align}
and using this, we get:
 \begin{align}
 \frac{1}{2} | & (\bra{\Phi_0^{n-k}}\ket{\psi_0})^2 + (\bra{\Phi_1^{n-k}}\ket{\psi_1})^2 \nonumber \\
 & - 2 \bra{\Phi_0^{n-k}}\ket{\psi_0} \bra{\Phi_1^{n-k}}\ket{\psi_1} |  \geq \sqrt{1 - \epsilon^2}.
 \end{align}
Using the triangle inequality, we have:
 \begin{align}
  \frac{1}{2} \Big[ \abs{\bra{\Phi_0^{n-k}}\ket{\psi_0}}^2 & + \abs{\bra{\Phi_1^{n-k}}\ket{\psi_1}}^2  \Big]  \geq \sqrt{1 - \epsilon^2 }.
\end{align}
Using the Cauchy-Schwarz inequality, we have:
\begin{align}
\bra{\psi_0}\ket{\psi_0} + \bra{\psi_1}\ket{\psi_1} & \geq \abs{\bra{\Phi_0^{n-k}}\ket{\psi_0}}^2  + \abs{\bra{\Phi_1^{n-k}}\ket{\psi_1}}^2 \\
& \geq \sqrt{1 - \epsilon^2}.
\end{align}
Since the overall state $\ket{\Psi'}$ is normalized, we have $|\bra{\chi'}\ket{\chi''} |\leq 1 - \sqrt{1 - \epsilon^2}$. Thus, we get our expression for fidelity as:
 \begin{align}
 \nonumber F(\ket{\Psi_i}, \ket{\Psi_j}) & =  \abs{\bra{\psi_0}\ket{\psi_0} + \bra{\psi_1}\ket{\psi_1} + \bra{\chi'}\ket{\chi''}}^2 \\
\nonumber  &=\left(\abs{\bra{\psi_0}\ket{\psi_0}|-| \bra{\psi_1}\ket{\psi_1} + \bra{\chi'}\ket{\chi''}}\right)^2  \\
 & \geq 1 - \epsilon^2-\epsilon^4=1-\Tilde{\epsilon}^2,
 \end{align}
 where $\Tilde{\epsilon}=\sqrt{\epsilon^2+\epsilon^4}$.
 
% which gives $D(\ket{\Psi_i}, \ket{\Psi_j}) \leq \epsilon$.
\end{proof}

%%%%%%%%%%%%%
\iffalse
\fede{Alternative take:
\begin{proof}
The malicious agents can try to modify a portion of the state in order to violate the anonymity of the voters. In the ideal case, the shared entangled state and all the operations are perfectly symmetrical, so there is no way the colluding agents can distinguish between any two voters and the protocol would be perfectly secure. In the non-ideal case, the dishonest players can exploit the imperfections of the state to create an asymmetric state $\ket{\psi_i}$ in such a way that when the honest agent $i$ is voting then they become distinguishable from the other honest agents. We know from theorem \ref{theo1} that we can neglect the probability of using a state that is larger than $\epsilon$ in trace distance from an ideal GHZ state. As a consequence, we have
\begin{equation}
    D(\ket{\psi_i},\ket{\text{GHZ}})\geq \epsilon
\end{equation}
By using the triangle inequality, we can upper-bound the trace distance between the states $\ket{\psi_i}$ and $\ket{\psi_j}$, that could violate the anonymity of voters $i$ and $j$ respectively
\begin{equation}
    D(\ket{\psi_i},\ket{\psi_j})\leq D(\ket{\psi_i},\ket{\text{GHZ}})+D(\ket{\psi_j},\ket{\text{GHZ}})=2\epsilon.
\end{equation}
\end{proof}
}
\fi
%%%%%%%%%%%%%%

We are now ready to prove Theorem \ref{theo2}.

\begin{manualtheorem}{2}
At any round $\ell \in [N]$ with voting agent $k$ (who has unique index $\omega_k = \ell$), if the agents use a state $\ket{\psi}$ such that $F(\ket{\psi}, \ket{GHZ})\geq \sqrt{1-\epsilon^2}$ to perform \textsf{Voting}, then for the optimal strategy that any subset of malicious agents $\mathcal{D}$ can use to guess the identity of the voting agent $k$ correctly, we have
\begin{equation}\label{eq:anon}
    \forall j \in W_H,\: \Pr[ \mathcal{D} \emph{ guess } j] = \begin{cases}
\tfrac{1}{H}+\Tilde{\epsilon} & \text{for $j=k$}\\
\tfrac{1-\Tilde{\epsilon}}{H} & \text{for $j \neq k$},
\end{cases}
\end{equation}
where $W_H$ is the set and $H$ the number of honest agents.
\label{th:anon}
\end{manualtheorem}

\begin{proof}
We will now show that if the agents share close to the GHZ state, then the voter remains anonymous.
From Theorem \ref{theo1}, we saw that the probability that the state used for voting satisfies $F'(\ket{\Psi}) \leq \sqrt{1 - \epsilon^2}$ is given by $
\text{Pr}[C_\epsilon] \leq \eta$
for the honest agents, where $\eta$ depends on the number of runs of the verification protocol. Thus, by doing enough runs, we can make this very small, and so we have that the state used for voting will be close to the GHZ state, as given by $F'(\ket{\Psi}) \geq \sqrt{1 - \epsilon^2}$.

From the previous proof, we see that if $F'(\ket{\Psi}) \geq \sqrt{1-\epsilon^2}$, the distance between the states if agent $i$ or $j$ was the voter is $D(\ket{\Psi_i}, \ket{\Psi_j}) \leq \Tilde{\epsilon}$. A malicious agent who wishes to guess the identity of the voter would make some sort of measurement to do so. Thus, we wish to find the maximum success probability of a measurement that could distinguish between the $H$ states that are the result of the voter (who can only be an honest agent) applying the $\sigma_x \sigma_z$ transformation.

The success probability of discriminating between $H$ states is given by $ \sum_{i=1}^H p_i \text{Tr} (\Pi_i \rho_i)$. From Lemma \ref{l:dishonest}, we know that the distance between any two states after the voter's transformation is upper-bounded by $\Tilde{\epsilon}$. Thus, if we take $\ket{\alpha} = \ket{\Psi_j}$, then we know that any of these $H$ states is of distance $\Tilde{\epsilon}$ away from this same state $\ket{\alpha}$.

For any POVM element $P$, we can write the trace distance between two states $\rho, \sigma$ as $
 \text{Tr} \big[ P (\rho - \sigma) \big] \leq D(\rho, \sigma)$.
Thus, we have for a POVM element $\Pi_i$ and for states $\ket{\Psi_i}, \ket{\alpha}$:
\begin{align}
 \text{Tr} (\Pi_i \ket{\Psi_i}\bra{\Psi_i}) - \text{Tr} (\Pi_i \ket{\alpha}\bra{\alpha})  \leq \Tilde{\epsilon}.
\end{align}
Assuming that each honest agent has an equal chance of becoming the voter, the probability that the malicious agents can guess the identity of the voter is bounded by:
\begin{align}
\text{Pr}[\text{guess}]
& = \sum_{i=1}^H \frac{1}{H} \text{Tr} (\Pi_i \ket{\Psi_i}\bra{\Psi_i})
\\
& \leq \frac{1}{H} \sum_{i=1}^H \Big[ \text{Tr} (\Pi_i \ket{\alpha}\bra{\alpha}) + \Tilde{\epsilon}
 \Big]  \\
& = \frac{1}{H} \text{Tr}\Big[\sum_{i=1}^H \Pi_i \ket{\alpha}\bra{\alpha}\Big] + \frac{1}{H} H  \Tilde{\epsilon}  \\
& = \frac{1}{H} \text{Tr} (\ket{\alpha}\bra{\alpha}) + \Tilde{\epsilon} \\
& = \frac{1}{H} + \Tilde{\epsilon}.
\end{align}

As we said, the probabilities for the other agents come from the fact that all the agents that are not voting perform exactly the same transformation on the state, so it is impossible for the dishonest parties to distinguish between them, hence the probability of guessing their identity is the same.

\end{proof}

\section{Appendix D: Proof of Theorem 3} \label{appendixTheorem3}

\begin{manualtheorem}{3}
If at round $\ell$ the agents are honest and use a state $\ket{\psi}$ such that $F(\ket{\psi}, \ket{GHZ})\geq \sqrt{1-\epsilon^2}$ to perform \textsf{Voting}, then the probability that there is an error in the tally in the $\ell$-th round is upper bounded by $\epsilon$,
\begin{equation}
P^{er}_{\ell}\leq \epsilon.
\end{equation}
\end{manualtheorem}

\begin{proof}
At each round, only one vote is declared. The state $\ket{\psi}$ maximizing this probability can be at most $\epsilon$-far in trace distance. Knowing that $Tr[\Pi(\rho-\tau)]\leq D(\rho,\tau)$ for any POVM $\Pi$, the probability of having one error using the state $\rho=\ket{\psi}\bra{\psi}$, instead of the correct state $\tau=\ket{GHZ}\bra{GHZ}$ is
\[
P^{er}_{\ell}=Tr[\Pi_{\ell}\mathcal{H}^{\otimes^N}\rho]\leq Tr[\Pi_{\ell}\mathcal{H}^{\otimes^N}\tau] + D(\rho,\tau)=\epsilon,
\]
where $\Pi_{\ell}$ is some operator that evaluates the distance from the correct $\ell$-th output of the state measured in the Hadamard basis, $\mathcal{H}^{\otimes^N}$ is the product of local Hadamard applied by each voter  and we used the fact that if we measure the correct state it is impossible to have an error.
\end{proof}

\section{Appendix E: Additional candidates} \label{appendixAdditionalCandidates}
If there are $K$ candidates, each candidate identifier will have $\log_2 K$ digits and each election can provide the preference for at most 1 digit of each voter. If we repeat the whole protocol $\log_2 K$ times, keeping the same secret index for each voter at all times, we end up with a greater election votes vector $\mathbf{\underline{E}}=\mathbf{E}^{(1)}\mathbf{E}^{(2)}...\mathbf{E}^{(\log_2 K)}$ formed by the election vote vector of each election by summing the row of the corresponding bulletin board $\mathbf{E}^{(i)}=\sum_k \mathbf{B}^{(i)}$. So the sub-election 1 will result in a vector $\mathbf{E}^{(1)}=e^{(1)}_{\omega_1} e^{(1)}_{\omega_2} ... e^{(1)}_{\omega_N}$, where  $e^{(1)}_{\omega_k}$ is the value of the first digit of the preference of voter $k$, with secret index $\omega_k$, and so on for all the other sub-elections.
If we want to perform an election with 3 candidates and 7 voters, allowing also the possibility of a null vote, which will be candidate $(0,0)$, we need to carry out 2 sub-elections, and result in the following table:
\[
\mathbf{\underline{E}}=\begin{pmatrix}
0 & 0\\
1 & 0 \\
1 & 1\\
0 & 0\\
0 & 0\\
1 & 1 \\
0 & 1
\end{pmatrix},
\]
where two agents voted for candidate 3, candidates 1 and 2 received one vote each and the rest of the voters decided not to express a preference.

\end{document}